\documentclass{amsart}
\usepackage{amsmath, amssymb, fullpage, amsthm, mathrsfs, verbatim, epsf, graphicx, multirow, array, color}
\usepackage[all]{xy}

\newcommand{\Q}{\mathbb{Q}}

\newcommand{\R}{\mathbb{R}}
\newcommand{\N}{\mathbb{N}}

\newcommand{\mcN}{\mathcal{N}}

\newcommand{\id}{\mathbb{I}}
\newcommand{\h}{\mathcal{H}}
\newcommand{\A}{\mathcal{A}}
\newcommand{\tbx}{\textbf{x}}
\newcommand{\tbf}{\textbf{f}}

\newcommand{\D}{\partial}
\newcommand{\ds}{\frac{d}{ds}}
\newcommand{\dsi}[1]{\frac{d^{#1}}{ds^{#1}}}
\newcommand{\Spec}{\textrm{Spec }}

\def\ba #1\ea{\begin{align} #1 \end{align}}
\def\bas #1\eas{\begin{align*} #1 \end{align*}}
\def\bml #1\eml{\begin{multline} #1 \end{multline}}
\def\bmls #1\emls{\begin{multline*} #1 \end{multline*}}

\def\acl{\blacktriangleright\hspace{-4pt}\vartriangleleft}

\newtheorem{thm}{Theorem}
\newtheorem{lem}{Lemma}
\newtheorem{cor}{Corollary}

\theoremstyle{remark}
\newtheorem{eg}{Example}

\theoremstyle{definition}
\newtheorem{dfn}{Definition}
\newtheorem{rem}{Remark}

\begin{document}
\author{Susama Agarwala}
\address{Oxford University \\
Mathematical Institute\\
Radcliff Observatory Quarter\\
Woodstock Road\\
Oxford \\
UK}

\author{Colleen Delaney}
\address{University of California Santa Barbara \\
South Hall, Room 6607\\
University of California,\\
Santa Barbara, CA 93106\\
USA}

\title{Generalizing the Connes Moscovici Hopf algebra to contain all
  rooted trees}

\begin{abstract}
This paper defines a generalization of the Connes-Moscovici Hopf
algebra, $\h(1)$ that contains the entire Hopf algebra of rooted
trees. A relationship between the former, a much studied object in
non-commutative geometry, and the latter, a much studied object in
perturbative Quantum Field Theory, has been established by Connes and
Kreimer. The results of this paper open the door to study the
cohomology of the Hopf algebra of rooted trees.
\end{abstract}

\maketitle



Connes and Moscovici define a family of Hopf algebras, $\h(n)$, on
$n$-dimensional flat manifolds in \cite{CMos98}. The action of these
Hopf algebras on the group of smooth functions on the frame
bundle, crossed with local diffeomorphism on $M$, is a well studied
object in non-commutative geometry.

In \cite{CK98}, Connes and Kreimer delineate a remarkable set of
similarities between the Hopf algebra of rooted trees, $\h_{rt}$,
which is important to understanding the perturbative regularization of Quantum
Field Theories (QFTs), and $\h(1)$. The Hopf algebra $\h(1)$ is defined
by a Lie algebra generated by two vector fields on the orientation
preserving frame bundle $F^+M$, $X$ and $Y$, and a set of linear
operators $ \{\delta_i | i \in \N\}$. The linear operators correspond
to a sub Hopf algebra of $\h_{rt}$, the vector field $Y$ corresponds
to the grading operator on $\h_{rt}$ and the vector field $X$ to the
natural growth operator on $\h_{rt}$.

While a lot is known about the cohomologies of $\h(n)$, little is
explicitly known about the cohomology classes of $\h_{rt}$. In
particular, since the Lie algebra associated to the Lie group $\Spec
\h_{rt}$ is freely generated, the Hochschild cohomologies
$HH^\bullet(\h_{rt}) = 0$ for $\bullet >1$. There is a one to one
correspondence between generators of $HH^1(\h_{rt})$ and primitive
elements of $\h_{rt}$ \cite{BK06}. Many primitive elements are known,
but only one element of $HH^1\h_{rt}$.

The one known Hochschild one co-cycle, the grafting
operator $B_+$, however, plays an important role in understanding the
combinatorics underlying the Dyson Schwinger equations for
perturbative QFTs \cite{BK06, K06dse, KY06}. Knowing more about
cohomology of $\h_{rt}$ will lead to a better understanding of the
process of perturbative regularization of QFTs, as well as the deep
connection between non-commutative geometry and renormalization of the
same.  Framing the full Hopf algebra of rooted trees in the context of
$\h(1)$ is a first step in this direction.

This paper introduces a generalization of $\h(1)$, $\h_{rt}(1)$, that
contains all of $\h_{rt}$ as a sub Hopf algebra. Key to this
construction is an observation by Cayley in 1881 relating rooted
trees to a particular first order differential equation. We use this
to generalize the natural growth operator on $\h_{rt}$, and thus to
define a family of vector fields $X_t$ on $F^+M$, corresponding to the
generalized natural growth operators. The commutation of these vector
fields $X_t$ with the element $\delta_1$ in $\h(1)$ gives a family of
linear operators $\delta_t$, each corresponding to a rooted tree $t
\in \h_{rt}$. In order to build this Hopf algebra $\h_{rt}(1)$, we
have to relax the requirement that the manifold $M$ is flat.

This paper is organized as follows. Section \ref{treesanddifs} recalls
the construction of Cayley, and the relationship between differentials
and rooted trees. Section \ref{hrt} recalls the definition of
$\h_{rt}$, and the natural growth operator. We generalize the natural
growth operator, and show that all of $\h_{rt}$ can be generated by
repeated application of the generalized natural growth operator to the
single vertex tree. Section \ref{CMpict} recalls the definition of
$\h(1)$ and constructs the generalized Hopf algebra $\h_{rt}(1)$.

\section{A relationship between rooted trees and differentials \label{treesanddifs}}

A relationship between rooted trees and differentials has been known
since Cayley's work in 1881 \cite{Cayley}. This relationship is
fundamental to Butcher's work to solve differential equations using
the Runge Kutta method \cite{Butcherbook}. For a brief but clear
exposition of the results in this section, see
\cite{BrouderRenorm}. Let $\tbx(s)$ be a vector field and \ba
\ds\tbx(s) = \tbf(\tbx) \quad \; \quad \tbx(0) = \tbx_0
\;.\label{ODE}\ea the system of differential equation of interest to
Butcher. The Butcher group for this system of differential equations
is the Lie group defining the Hopf algebra of rooted
trees. Specifically, let $\h_{rt}$ be the Hopf algebra of rooted
trees, defined in section \ref{hrt} of this paper. Let $G$ be the
Butcher group for this differential equation. Then $\h_{rt}$ is the
graded dual of the universal enveloping algebra of the Lie algebra of
$G$ \bas \h_{rt} = \mathcal{U}^\vee(\textrm{Lie}(G)) \;. \eas The
group $G$ defines the combinatorics at the heart of $\h_{rt}$. For
more details, see \cite{Dur}. For a detailed exposition on the
relationship between the Butcher group and the Hopf algebra of rooted
trees, see \cite{BrouderRenorm}.

For the rest of this paper, we consider differential equations of the
form \eqref{ODE}. For now, let $\tbx(s)$ be a smooth vector field \bas
\tbx(s): \R \rightarrow \R^n\;. \eas In section \ref{CMpict}, we will
generalize these arguments from $\R^n$ to more general manifolds.

\begin{dfn}
Write $\D_i = \frac{\partial}{\partial x^i}$. We use the shorthand
$\D_{i_1\ldots i_k}f(\tbx) := \D_{i_1}\D_{i_2}\ldots \D_{i_k}f(\tbx)$.
\end{dfn}

In the context where $\tbx(s)$ is a vector field on $\R^n$, the
following correspondence exists between derivations of $\tbx(s)$ and
non-planar rooted trees.

\begin{dfn}
Rooted trees are non-planar simply connected, oriented graphs with a marked
vertex, called the root. All the edges of $t$ are oriented away from
the root.
\end{dfn}

In this case, rooted trees are drawn root vertex up, indicated by a
circle around the root vertex.

Cayley associates sums of rooted trees to the derivative
$\dsi{n}\tbx(s)$ as follows. \begin{enumerate}\item The expression
  $\dsi{n}\tbx(s)$ is written in terms of sums of products of
  functions of the form $\D_{i_1\ldots i_k}f^j(\tbx)$ by the chain
  rule and equation \eqref{ODE}. Each such summand has a total of
  $n-1$ derivatives. \item Each appearance of the term $f^j(\tbx)$
  corresponds to a vertex of the rooted tree. \item Each appearance of
  a $\D_i$ corresponds to an edge flowing from its associated vertex,
  connecting the vertex the derivative acts upon to another vertex
  indexed by the same $i$.\end{enumerate}

The table below gives a few examples of this procedure.

\begin{tabular}{m{6cm}|c}
$\ds\tbx(s) = \tbf(\tbx)$ & $\xy \POS(0,0) *+{\bullet} *\cir{}
  \POS(3,0) *+{f^j}\endxy$ \\ $\dsi{2}\tbx(s) = \ds\tbf(\tbx) = D
  \tbf(\tbx) \cdot \ds \tbx(t) = \D_i f^j(\tbx) \ds x^i(s) $ & $\xy
  \POS(0,4) *+{\bullet} *\cir{} \POS(6,4) *+{\D_if^j} \POS(0,4)
  \ar@{-} + (0, -8) \POS(0, -4) *+{\bullet} \POS(2, -4) *+{f^i} \endxy
  $\\ $\dsi{3}\tbx(s) = \ds( \D_i f^j(\tbx)) f^i + \D_if^j \ds f^i =
  (\D_{ik}f^j) f^kf^i + (\D_if^j)(\D_kf^i) f^k $ & $\xy \POS(0,4)
  *+{\bullet} *\cir{} \POS(6,4) *+{\D_{ik}f^j} \POS(0,4) \ar@{-} +
  (-4, -8) \POS(-4, -4) *+{\bullet} \POS(-6, -4) *+{f^i} \POS(0,4)
  \ar@{-} + (4, -8) \POS(4, -4) *+{\bullet} \POS(6, -4) *+{f^k} \endxy
  + \xy \POS(0,8) *+{\bullet} *\cir{} \POS(6,8) *+{\D_if^j} \POS(0,8)
  \ar@{-} + (0, -8) \POS(0, 0) *+{\bullet} \POS(6, 0) *+{\D_kf^i}
  \POS(0,0) \ar@{-} + (0, -8) \POS(0, -8) *+{\bullet} \POS(2, -8)
  *+{f^k}\endxy $ \\ $\dsi{4}\tbx(s) =(\D_{ikl}f^j) f^lf^kf^i +
  2(\D_{ik}f^j)(\D_lf^i)f^lf^k + (\D_{li}f^j)(\D_kf^i)f^l f^k +
  (\D_if^j)(\D_{lk}f^i) f^lf^k + (\D_if^j)(\D_kf^i) (\D_lf^k)f^l$ &
  $\xy \POS(0,4) *+{\bullet} *\cir{} \POS(6,4) *+{\D_{ikl}f^j}
  \POS(0,4) \ar@{-} + (-4, -8) \POS(-4, -4) *+{\bullet} \POS(-6, -4)
  *+{f^i} \POS(0,4) \ar@{-} + (4, -8) \POS(4, -4) *+{\bullet} \POS(6,
  -4) *+{f^k} \POS(0,4) \ar@{-} + (0, -8) \POS(0, -4) *+{\bullet}
  \POS(-2, -4) *+{f^l} \endxy + 3 \xy \POS(0,8) *+{\bullet} *\cir{}
  \POS(6,8) *+{\D_{ik}f^j} \POS(0,8) \ar@{-} + (-4, -8) \POS(-4, 0)
  *+{\bullet} \POS(-8, 0) *+{\D_lf^i} \POS(-4,0) \ar@{-} + (0, -8)
  \POS(-4, -8) *+{\bullet} \POS(-6, -8) *+{f^l} \POS(0,8) \ar@{-} +
  (4, -8) \POS(4, -0) *+{\bullet} \POS(6, 0) *+{f^k} \endxy + \xy
  \POS(0,8) *+{\bullet} *\cir{} \POS(6,8) *+{\D_if^j} \POS(0,8)
  \ar@{-} + (0, -8) \POS(0, 0) *+{\bullet} \POS(6, 0) *+{\D_{lk}f^i}
  \POS(0,0) \ar@{-} + (-4, -8) \POS(-4, -8) *+{\bullet} \POS(-6, -8)
  *+{f^k} \POS(0,0) \ar@{-} + (4, -8) \POS(4, -8) *+{\bullet} \POS(6,
  -8) *+{f^l}\endxy + \xy \POS(0,12) *+{\bullet} *\cir{} \POS(6,12)
  *+{\D_if^j} \POS(0,12) \ar@{-} + (0, -8) \POS(0, 4) *+{\bullet}
  \POS(6, 4) *+{\D_kf^i} \POS(0,4) \ar@{-} + (0, -8) \POS(0, -4)
  *+{\bullet} \POS(6, -4) *+{\D_lf^k} \POS(0,-4) \ar@{-} + (0, -8) \POS(0, -12)
  *+{\bullet} \POS(2, -12) *+{f^l} \endxy $

\end{tabular}

Notice that the $n^{th}$ derivative of $\tbx(s)$ with respect to $s$
gives rise to a linear combination of trees with $n$ vertices. In fact, it
is a linear combination of all the trees with $n$ vertices.

For a fixed differential equation as in \eqref{ODE} above, following
Butcher \cite{Butcherbook}, one can associate to any rooted tree $t$ a
function of $f$ and its derivatives, \bas \phi : \h_{rt} &\rightarrow
C^\infty(\R^n, \R^n) \;. \eas In order to make this map
explicit, we define some notation.

\begin{dfn}
\begin{enumerate}
\item Let $V(t)$ be the vertex set of a tree, and $E(t)$ the set of
edges.
\item The \emph{fertility} of a vertex, $v$, is the number of edges for
which $v$ is the initial vertex. Write this $fert(v)$.
 \item If there is an edge flowing from a vertex $v$ to a vertex $w$,
  we say that $w$ is the daughter of $v$, or that $v$ is the parent of $w$.
\item The only vertex on a connected rooted tree without a parent is
  the root. Any vertex with fertility $0$ are leaf vertices.
\end{enumerate}
\end{dfn}

To fix notation, define $\ds x^i = f^i = \phi^i(\bullet)$. The map
$\phi$ is defined componentwise on rooted trees. To calculate
$\phi^j(t)$, assign an index $i_v$ to each vertex $v \in V(t)$. If $r$
is the root vertex of $t$, then $i_r = j$. Then
\ba \phi^j(t) = \prod_{v\in V(t)} \left( \prod_{w \textrm{ daughter of } v} \D_{i_w} \right) \phi^{i_v}(\bullet) \label{phimap} \;.\ea For instance,
given a tree \bas t = 
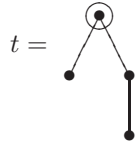
\xy \POS(0,4) *+{\bullet} *\cir{} \POS(0,4)
\ar@{-} + (-4, -8) \POS(-4, -4) *+{\bullet} \POS(0,4) \ar@{-} + (4,
-8) \POS(4, -4) *+{\bullet} \POS(4, -4) \ar@{-} + (0, -8) \POS(4, -12)
*+{\bullet} \endxy\eas one has \bas \phi^k(t) =
 (\D_{ij}\phi^k(\bullet)) (\D_l\phi^i(\bullet)) \phi(\bullet)^j\phi(\bullet)^l \;. \eas

Given a fixed differential equation, one can view the map $\phi$ as
assigning a differential operator to the tree $t$.

\begin{dfn}
Let $t$ be a rooted tree, with root vertex $r$. Write \bas \phi_t =
\left[\prod_{v\in V(t)\setminus r} \left( \prod_{w \textrm{ daughter
      of } v} \D_{i_w} \right) \phi^{i_v}(\bullet)\right]\prod_{w
  \textrm{ daughter of } r } \D_{i_w}\;. \eas
\label{phidifop}
\end{dfn}

This identifies a family of differential operators indexed by rooted
trees. Each of these differential operators defines a linear map on
smooth functions in the obvious way.

In particular, in this notation, $\phi(t)^i = \phi_t
(\phi(\bullet)^i)$. Later in the paper, we replace the function $f^j$
decorating the root vertex of $\phi(t)^j$ with a different function,
and consider functions of the form \ba \phi_t(h) =\left[\prod_{v\in
    V(t)\setminus r} \left( \prod_{w \textrm{ daughter of } v}
  \D_{i_w} \right) \phi^{i_v}(\bullet)\right]\prod_{w \textrm{
    daughter of } r } \D_{i_w} h \label{phimapgen} \; ,\ea for $h \in
C^\infty(\R^n)$. In particular, $\phi_\bullet(h) = h$.

\section{The Hopf algebra of rooted trees \label{hrt}}

Let $\h_{rt}$ be the Hopf algebra of rooted trees. This Hopf algebra
is defined in detail in \cite{CK98} and summarized here. Consider the
vector field generated by all rooted trees $t$ \bas \Q \langle t | t
\textrm{ rooted tree} \rangle \;. \eas This is a unital algebra with
the unit representing the empty tree, $1 = \id_{\h_{rt}} =
t_\emptyset$. Multiplication on this algebra is commutative, and
corresponds to the disjoint union of trees. The product, or disjoint union, of non-trivial trees is called a forest.

The algebra $\h_{rt}$ is graded by the number of vertices in the
tree \bas \h_{rt} = \bigoplus_{i=0}^\infty \h_{rt,i} ;\quad \h_{rt,0}
= \Q ;\quad \h_{rt, i} = \Q \langle \{t| t \textrm{ tree or forest
  with $i$ vertices} \} \rangle \;. \eas

There is a grading operator $Y$ on $\h_{rt}$ such that for a single
tree with $|V(t)| = n$ vertices, $t \in \h_{rt,n}$, $Y(t)= nt$. Before
defining the coproduct structure on $\h_{rt}$, we define a few terms.

\begin{dfn}
A proper admissible cut of a tree $t$ is a collection of edges, $c \in
E(t)$, such that any path from a root vertex to a leaf vertex contains
at most one such edge.
\end{dfn}

Removing the edges corresponding to an admissible cut of $t$ gives
a set of subtrees of $t$.

\begin{dfn}
The subtree of $t$ associated to a proper admissible cut $c$ that
contains the root vertex of $t$ is the root subtree defined by $c$,
$R_c(t)$. The other set of subtrees is the pruned forest defined by $c$,
$P_c(t) = \prod_i P_{c,i}(t)$. Each $P_{c,i}(t)$ is a single
tree in this forest. The root vertex of each $P_{c,i}(t)$ is defined
by the orientation of $t$.
\end{dfn}

The set of admissible cuts of a tree consists of the proper admissible
cuts and two trivial cuts, one such that $R_c(t) = t$ (called the
\emph{empty cut}), and one such that $P_c(t) = t$ (called the
\emph{full cut}).

The bialgebra structure on $\h_{rt}$ is given by a coproduct \bas
\Delta(t) = \sum_{c \textrm{ admis. cut}}P_c(t) \otimes R_c(t)\;. \eas
This coproduct is compatible with multiplication on $\h_{rt}$. Let $c$
be an admissible cut of $t$ and $c'$ an admissible cut of $t'$. Then
\bas \Delta(t t') = \sum_{c \textrm{ admis. cut}} \sum_{c' \textrm{
    admis. cut}} P_c(t) P_{c'}(t') \otimes R_c(t) R_{c'}(t) \;. \eas
This is a grading preserving coassociative coproduct. The co-unit is given
by \bas \varepsilon(t) = \begin{cases} t & \hbox{$t \in \h_{rt,0}$;}
  \\ 0 & \hbox{else.} \end{cases}\eas

Thus defined, $\h_{rt}$ is a connected graded bialgebra over a
commutative ring, $\Q$. Therefore, it is a Hopf algebra. The antipode
is given by \bas S(t) = -t -\sum_{c \textrm{ proper admis. cut}}
P_c(t) S(R_c(t))\;. \eas

The map $\phi$ defined in \eqref{phimap} is an algebra
homomorphism. For a forest $t t'$, \bas \phi^i(tt') = \phi^i(t)
\phi^i(t') \;.\eas Furthermore \bas \Delta \phi^i(t) = (\phi^i \otimes
\phi^i) \Delta t \;. \eas

\subsection{Two operators on $\h_{rt}$\label{NB}}

In \cite{CK98} the authors identify two operators on $\h_{rt}$ of importance to the analysis in this paper. The first, called a natural growth operator, denoted $N$, is a derivation on $\h_{rt}$. The other, called a grafting operator, denoted $B_+$, is linear operator on $\h_{rt}$, and a Hochschild one-cocycle \cite{BroadK}. In this section, we review these two operators.

In \cite{CK98}, the authors define a sub Hopf algebra of $\h_{rt}$
generated by elements formed by repeated application of the natural
growth operator on the single vertex tree.

\begin{dfn}
The natural growth operator \bas N : \h_{rt, i} \rightarrow \h_{rt,
  i+1}\eas maps a tree, $t$, to a sum of trees, $N(t)$, formed by
increasing the fertility of each vertex of $t$ by one. Each tree in the sum
has one more vertex than $t$. \end{dfn}

For instance, \bas N(
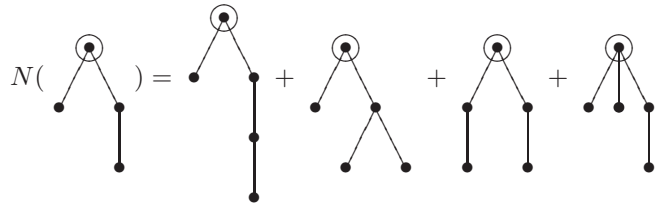
\xy \POS(0,4) *+{\bullet} *\cir{} \POS(0,4)
\ar@{-} + (-4, -8) \POS(-4, -4) *+{\bullet} \POS(0,4) \ar@{-} + (4,
-8) \POS(4, -4) *+{\bullet} \POS(4, -4) \ar@{-} + (0, -8) \POS(4, -12)
*+{\bullet} \endxy ) = \xy \POS(0,8) *+{\bullet} *\cir{} \POS(0,8)
\ar@{-} + (-4, -8) \POS(-4, 0) *+{\bullet} \POS(0,8) \ar@{-} + (4, -8)
\POS(4, 0) *+{\bullet} \POS(4, 0) \ar@{-} + (0, -8) \POS(4, -8)
*+{\bullet} \POS(4, -8) \ar@{-} + (0, -8) \POS(4, -16) *+{\bullet}
\endxy + \xy \POS(0,4) *+{\bullet} *\cir{} \POS(0,4) \ar@{-} + (-4,
-8) \POS(-4, -4) *+{\bullet} \POS(0,4) \ar@{-} + (4, -8) \POS(4, -4)
*+{\bullet} \POS(4, -4) \ar@{-} + (4, -8) \POS(8, -12) *+{\bullet}
\POS(4, -4) \ar@{-} + (-4, -8) \POS(0, -12) *+{\bullet} \endxy+  \xy \POS(0,4) *+{\bullet} *\cir{} \POS(0,4) \ar@{-} + (-4,
-8) \POS(-4, -4) *+{\bullet} \POS(0,4) \ar@{-} + (4, -8) \POS(4, -4)
*+{\bullet} \POS(4, -4) \ar@{-} + (0, -8) \POS(4, -12) *+{\bullet}
\POS(-4, -4) \ar@{-} + (0, -8) \POS(-4, -12) *+{\bullet} \endxy + \xy
\POS(0,4) *+{\bullet} *\cir{} \POS(0,4) \ar@{-} + (-4, -8) \POS(-4,
-4) *+{\bullet} \POS(0, 4) \ar@{-} + (0, -8) \POS(0, -4) *+{\bullet}
*+{\bullet} \POS(0,4) \ar@{-} + (4, -8) \POS(4, -4) *+{\bullet}
\POS(4, -4) \ar@{-} + (0, -8) \POS(4, -12) *+{\bullet} \endxy \eas

This is a derivation on $\h_{rt}$. In particular, for $t$, $s$, two rooted tress, \bas N(t + s) = N(t) + N(s) \;.\eas It can be extended to act on forests: $N (t_1 \ldots t_n)$ gives a sum of forests defined by increasing the valence each vertex in the \emph{forest} by one. In particular, for $t$ and $s$ rooted trees, \bas N(ts) = N(t)s + tN(s) \;.\eas

Natural growth is closely related to $B_+$, the grafting operator. This linear operator maps
a forest to a single tree formed by connecting each root vertex of the
forest to a new root vertex. Write this operator \bas B_+ :
\h_{rt} \rightarrow \h_{rt}\;. \eas On a forest $t_1t_2$, this is
defined \bas B_+(t_1 t_2) = \xy \POS(0,4) *+{\bullet} *\cir{}
\POS(0,4) \ar@{-} + (-4, -8) \POS(0,4) \ar@{-} + (+4, -8) \POS(-4,-4)
*+{\bullet} \POS(4,-4) *+{\bullet} \POS(-5, -5) *+{t_1} \POS(6, -5)
*+{t_2}\endxy \;.\eas Specifically, if $t_1 = t_2 = \bullet$, then\bas
B_+(\bullet \bullet) = \xy \POS(0,4) *+{\bullet} *\cir{} \POS(0,4)
\ar@{-} + (-4, -8) \POS(0,4) \ar@{-} + (+4, -8) \POS(-4, -4)
*+{\bullet} \POS(4, -4) *+{\bullet}\endxy \;.\eas This can be extended
by linearity to all of $\h_{rt}$.

Notice that any rooted tree $t \in \h_{rt, m}$ can be written as a grafting of a product of trees $B_+(t_1\ldots t_n)$, with $t_1\ldots t_n \in \h_{rt, m-1}$. There is a unique admissible cut, $c$, of $t$ such that $R_c$ is a one valent tree. The argument of $B_+$ is the pruned forest defined by this cut, $P_c = t_1\ldots t_n$.

\begin{rem}
The operator $B_+$ is a Hochschild one-cocycle on $\h_{rt}$ \cite{BK06}. This is seen from the coproduct of the $B_+$ operator on forests: \ba \Delta B_+(t_1\ldots t_n) = (\id \otimes B_+) \Delta(t_1\ldots t_n) + B_+(t_1\ldots t_n) \otimes 1 \;.\label{Bcoprod}\ea
\end{rem}

Using the notation of section \ref{treesanddifs}, given an initial value
problem of the form \eqref{ODE}, we associate differential operators to the natural growth and grafting operators.

\begin{lem}
Given a fixed ODE, applying natural growth to a rooted tree, $t$, corresponds to taking a derivative by the parameter $s$ of the function $\phi^i(t)$:  \bas \phi^i(N(t)) = \ds
\phi^i(t)\;.\eas \label{stdnatgrowth}\end{lem}

\begin{proof}
From equation \eqref{phimap}, \bas \phi(t)^j =\prod_{v\in V(t)} \left(
\prod_{w \textrm{ daughter of } v} \D_{i_w} \right)
\phi(\bullet)^{i_v} \;.\eas Writing the operator $\ds =
\phi(\bullet)^i \D_i $, we see that \ba \ds \phi(t)^j =\phi(\bullet)^i
\D_i \left[\prod_{v\in V(t)} \left( \prod_{w \textrm{ daughter of } v}
  \D_{i_w} \right) \phi^{i_v}(\bullet)\right] \label{dsphit}\;.\ea

The natural growth operator gives an extra leaf daughter to each
vertex. Give these new vertices the index $u'$ if $u$ is the parent
vertx. Therefore, assign to each new leaf, the function
$\phi^{i_{v'}}$. Under this notation, \ba \phi^i(N(t)) =\sum_{u \in
  V(t)} \phi^{i_{u'}}(\bullet) \prod_{v\in V(t)}
\D_{i_{u'}}\left(\prod_{w \textrm{ daughter of } v} \D_{i_w}
\right)\phi^{i_v}(\bullet) \label{phiNt}\ea

Comparing equations \eqref{dsphit} and \eqref{phiNt} shows that the
latter is exactly the former after the product rule has been applied.

\end{proof}

As a result, write \ba \phi_N = \phi(\bullet)^i \D_i \;. \label{Ndiffop} \ea
In particular, this implies that \bas
\phi(\delta_k)^i = \dsi{k}x^i(s) \;.\eas

Unlike the case for natural growth, which is a derivation on $\h_{rt}$, $B_+$ is not a derivation on that Hopf Algebra. Therefore, there is not a natural derivation on the vector fields $\phi(t)^i$ defined by $B_+$. The linear operator, however, does define a simplification of Definition \ref{phidifop}.

For $t$ a rooted tree, write $t = B_+(t_1 \ldots t_n)$, \ba \phi^i(t) = \phi^i(B_+(t_1 \ldots t_n)) = (\prod_{j=1}^n \phi_{t_j})\phi^i(\bullet) \;. \ea Similarly, the operator $\phi_t$ from definition \ref{phidifop} can be written in terms of $B_+$: \ba
\phi_t = \prod_{j=1}^n \phi_{t_j}(\phi^{i_j}(\bullet) ) \prod_{j=1}^n
\D_{i_j} \;. \label{Bphinotation}\ea

Finally, we note the differential operator associated to $N_t(t')$, for $t$ and $t'$ rooted trees.

\begin{lem}
 Composing the differential operators $\phi_{N_t}$ and $\phi_{t'} $
 gives \bas \phi_{N_t} \phi_{t'} = \phi_{N_t(t')}
 \;.\eas \label{natgrowthgen}
\end{lem}

\begin{proof}
Let $t' = B_+(t_1\ldots t_n)$. From equations \eqref{phinatgrowth} and
\eqref{Bphinotation}, \bas \phi_{N_t} \phi_{t'} = \phi^k(t) \D_k
\left(\prod_{j= 1}^n \phi_t(\phi^{i_j}(\bullet)) (\prod_{j= 1}^n
\D_{i_j})\right)  = \\ \phi^k(t) \D_k (\prod_{j= 1}^n
\phi_t(\phi^{i_j}(\bullet))) (\prod_{j= 1}^n \D_{i_j})  + \phi^k(t)
(\prod_{j= 1}^n \phi_t(\phi^{i_j}(\bullet))) (\D_k \prod_{j= 1}^n
\D_{i_j}) \;.\eas By equation \eqref{NBrel} and Definition
\ref{phidifop}, the second line evaluates to \bas \sum_{i =
  1}^n\phi_{B_+(t_1\ldots N_t(t_i) \ldots t_n)} + \phi_{B_+(t t_1
  \ldots t_n)} = \phi_{N_t(t')}\;. \eas
\end{proof}

\subsection{Properties and generalizations of $N$ and $B_+$ \label{NBgen}}

On the level of trees, there is no reason to limit natural growth
only to growth by one vertex. One can grow a tree by any other tree. This realization is key to the analysis in this paper.

In this section, we develop this generalization of the natural growth operator, and explore some of its properties.

\begin{dfn}
For any tree $t \in \h_{rt}$, there is a generalized natural growth
operator \bas N_t : \h_{rt, i} \rightarrow \h_{rt,i + Y(t)} \eas that
maps a tree to a sum of trees, each formed by increasing the
fertility of a vertex by one. This additional edge connects
the vertex $v$ of the original tree to the root vertex of $t$. \end{dfn}

For instance, let \bas t = 
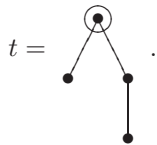
\xy \POS(0,4) *+{\bullet} *\cir{} \POS(0,4)
\ar@{-} + (-4, -8) \POS(-4, -4) *+{\bullet} \POS(0,4) \ar@{-} + (4,
-8) \POS(4, -4) *+{\bullet} \POS(4, -4) \ar@{-} + (0, -8) \POS(4, -12)
*+{\bullet} \endxy \; .\eas Then the tree \bas N_t(\delta_2) = \xy
\POS(0,12) *+{\bullet} *\cir{} \POS(0,12) \ar@{-} + (-4, -8) \POS(-4,
4) *+{\bullet} \POS(0,12) \ar@{-} + (4, -8) \POS(4, 4) *+{\bullet}
\POS(4,4) \ar@{-} + (-4, -8) \POS(0, -4) *+{\bullet} \POS(4,4) \ar@{-}
+ (4, -8) \POS(8, -4) *+{\bullet} \POS(8, -4) \ar@{-} + (0, -8)
\POS(8, -12) *+{\bullet} \endxy + \xy \POS (0,16) *+{\bullet} *\cir{}
\POS (0, 16) \ar@{-} + (0,-8) \POS(0,8) *+{\bullet} \POS (0, 8)
\ar@{-} + (0,-8) \POS(0,0) *+{\bullet} \POS(0,0) \ar@{-} + (-4, -8)
\POS(-4, -8) *+{\bullet} \POS(0,0) \ar@{-} + (4, -8) \POS(4, -8)
*+{\bullet} \POS(4, -8) \ar@{-} + (0, -8) \POS(4, -16) *+{\bullet}
\endxy \;.\eas In the first summand, the tree $t$ is hung off the root
vertex of $\delta_2$. In the second summand, it is hung off the unique
leaf vertex.

There are two special cases to be taken into consideration. Since $1
\in \h_{rt}$ corresponds to the empty tree with no vertices, \ba
N_t(1) = 0 \label{natgrowone}\;.\ea  This is as one expects, the derivation of a constant is $0$.  Similarly,
one may grow a tree $t$ by the empty tree. If one adds a trivial tree to each
vertex of $t$, this operation simply counts the number of vertices of
$t$. In other words, \ba N_1(t) = Y(t) \;. \label{natgrowbyone}\ea

The natural growth and the $B_+$ operator are related. For any tree
$t$ with root fertility $n$, there exists a product of trees $t_1 \ldots
t_n$ such that \bas t = B_+(t_1 \ldots t_{n})\;. \eas Applying
natural growth by $s$ gives \ba N_{s}(t) = B_+ (s t_1\ldots t_n) + \sum_{i=1}^n B_+ (t_1 \ldots N_s(t_i) \ldots t_n)
 \label{NBrel} \;.\ea The first summand corresponds to attaching $s$ to the root vertex of $t$. In particular, \ba N_t(\bullet) = B_+(t) \label{NBrel2} \;.\ea

Finally, we calculate the composition of the coproduct with the
generalized natural growth operator.

\begin{thm}
  Let $t, s$ be rooted trees. The coproduct of the operator $N_t$ is given by
  \bas \Delta  N_s(t) = (N_s\otimes \id) \Delta t + \sum_{c
    \textrm{ admis. cut}} (P_c(s) \otimes N_{R_c(s)}) \Delta t
  \;.\eas \label{Ntcoprod}
\end{thm}

\begin{proof}

Write $t = B_+(t_1\ldots t_n)$.  By equation \eqref{NBrel}, write
\bas  N_s(t) = B_+(st_1\ldots t_n) + \sum_{i=1}^n B_+ (t_1 \ldots N_s(t_i) \ldots t_n) \;.\eas Then by \eqref{Bcoprod}, \bmls \Delta N_s(t) = (\id \otimes B_+) \Delta(st_1\ldots t_n) + B_+(st_1\ldots t_n) \otimes 1 \\ + \sum_{i=1}^n (\id \otimes B_+) \Delta (t_1 \ldots N_s(t_i) \ldots t_n) + B_+ (t_1 \ldots N_s(t_i) \ldots t_n) \otimes 1 \;.\emls The first line comes from $\Delta B_+ (s t_1\ldots t_n)$, the second line from  $\Delta \sum_{i=1}^n B_+ (t_1 \ldots N_s(t_i) \ldots t_n)$. Collecting terms gives \ba \Delta N_s(t) = (\id \otimes B_+) \Delta(st_1\ldots t_n) + N_s(t) \otimes 1 + \sum_{i=1}^n (\id \otimes B_+) \Delta (t_1 \ldots N_s(t_i) \ldots t_n) \label{coprodexpand}\;.\ea The rest of this proof proceeds by induction on the number of vertices of $t$.

If $t = \bullet$, by \eqref{NBrel2}, \bas \Delta N_s(\bullet) = \Delta B_+(s) =(\id \otimes B_+)\Delta(s) + B_+(s) \otimes 1 \;.\eas Applying \eqref{NBrel2} again gives \bas \Delta N_s(\bullet) = \sum_{c
    \textrm{ admis. cut}}( P_c(s) \otimes N_{R_c(s)}) \Delta \bullet + (N_s \otimes \id) \Delta \bullet \;.\eas This proves the theorem for graphs of valence $1$.

Suppose the theorem holds for all trees with fewer than $k$ vertices. Consider $t$ to be a tree with $k$ vertices. Then the equation \eqref{coprodexpand} reads \bmls \Delta N_s(t) = N_s(t) \otimes 1 +  \sum_{c
    \textrm{ admis. cut}} P_c(s)P_c(t_1 \ldots t_n) \otimes B_+(R_c(s)R_c(t_1 \ldots t_n)) +   \\ \sum_{i=1}^n \left( P_c(t_1) \ldots  N_s(P_{c}(t_i)) \ldots P_c(t_n) \otimes B_+(R_c(t_1 \ldots t_n)) +  P_{c}(s) P_c(t_1 \ldots t_n) \otimes B_+(R_c(t_1) \ldots N_{R_{c}(s)}(t_i) \ldots \R_c(t_n))\right)\;.\emls
Combining the first summand with the third, and the second summand with the fourth (using equation \eqref{NBrel}), one gets \bas \Delta N_s(t) = (N_s \otimes \id ) \Delta t + \sum_{c \textrm{ admis. cut}} (P_c(s) \otimes N_{R_c(s)}) \Delta t \;.\eas This proves the theorem.

\end{proof}

As with natural growth by a vertex, the generalize natural growth operator on trees corresponds to a derivation.

\begin{lem}
In general, for $t$ and $t'$ rooted trees, \ba\phi^i(N_t(t')) =
\phi^j(t) \D_j(\phi^i(t'))  \label{gennatgrowth}\;.\ea
\label{gennatgrowthlem}\end{lem}

 \begin{proof}
This proof is a generalization of Lemma \ref{stdnatgrowth}, when $t =
\bullet$, and $\ds = \phi^j(\bullet) \D_j$.

 Natural growth by $t$ increases the fertility of each vertex by
 one. This is represented in the right hand side of
 \eqref{gennatgrowth} by differentiation with respect to
 $x^j(s)$. Instead of growing by a single vertex, $N_t$ grows by the
 tree $t$. This is represented by contraction with $\phi^j(t)$.
\end{proof}

\begin{cor}
Writing $Y(t) = N_1(t)$, as in \eqref{natgrowbyone}, and $\phi^i(1) = x^i(s)$, Lemma \ref{gennatgrowth} implies that \bas \phi^i(Y(t)) = x^j(s) \D_j (\phi^i(t) \;.\eas
\end{cor}

Just as in equation \eqref{phimapgen}, we associate a differential
operator to trees, in this paper, we also associate an operator to
(general) natural growth. In particular: \ba \phi_{N_t} = \phi^j(t)
\D_j\;. \label{phinatgrowth}\ea

\subsection{Generating $\h_{rt}$ using generalized natural growth operators}
In \cite{CK98}, the authors generated a sub Hopf algebra, $\h_1 \subset \h_{rt}$, of rooted trees using only natural growth by a single vertex. This Hopf algebra of rooted trees is defined by a set of generators \bas
N^{k-1}(\bullet) = \delta_k \;.\eas In loc. cit., the authors show that the algebra
$\Q[\{\delta_k | k \in \N\}]$ is a Hopf algebra.

Below are the trees for the
first few $\delta_k$. \bas \delta_1 = \xy \POS(0,4) *+{\bullet}
*\cir{} \endxy \eas \bas \delta_2 = \xy \POS(0,4) *+{\bullet} *\cir{}
\POS(0,4) \ar@{-} + (0, -8) \POS(0, -4) *+{\bullet} \endxy \eas \bas
\delta_3 = \xy \POS(0,4) *+{\bullet} *\cir{} \POS(0,4) \ar@{-} + (-4,
-8) \POS(-4, -4) *+{\bullet} \POS(0,4) \ar@{-} + (4, -8) \POS(4, -4)
*+{\bullet} \endxy + \xy \POS(0,8) *+{\bullet} *\cir{} \POS(0,8)
\ar@{-} + (0, -8) \POS(0, 0) *+{\bullet} \POS(0,0) \ar@{-} + (0, -8)
\POS(0, -8) *+{\bullet} \endxy \eas \bas \delta_4 = \xy \POS(0,4)
*+{\bullet} *\cir{} \POS(0,4) \ar@{-} + (-4, -8) \POS(-4, -4)
*+{\bullet} \POS(0,4) \ar@{-} + (4, -8) \POS(4, -4) *+{\bullet}
\POS(0,4) \ar@{-} + (0, -8) \POS(0, -4) *+{\bullet} \endxy + 3 \xy
\POS(0,8) *+{\bullet} *\cir{} \POS(0,8) \ar@{-} + (-4, -8) \POS(-4, 0)
*+{\bullet} \POS(-8, 0) \POS(-4,0) \ar@{-} + (0, -8) \POS(-4, -8)
*+{\bullet} \POS(0,8) \ar@{-} + (4, -8) \POS(4, -0) *+{\bullet} \endxy
+ \xy \POS(0,8) *+{\bullet} *\cir{} \POS(0,8) \ar@{-} + (0, -8)
\POS(0, 0) *+{\bullet} \POS(0,0) \ar@{-} + (-4, -8) \POS(-4, -8)
*+{\bullet} \POS(0,0) \ar@{-} + (4, -8) \POS(4, -8) *+{\bullet} \endxy
+ \xy \POS(0,12) *+{\bullet} *\cir{} \POS(0,12) \ar@{-} + (0, -8)
\POS(0, 4) *+{\bullet} \POS(0,4) \ar@{-} + (0, -8) \POS(0, -4)
*+{\bullet} \POS(0,-4) \ar@{-} + (0, -8) \POS(0, -12) *+{\bullet}
\endxy\eas

\begin{dfn}
Let $\h_{CK}= \Q [\{\delta_k | k \in \N\}]$ be the Hopf algebra
defined in \cite{CK98}, generated by the terms $\delta_k$. \end{dfn}

In this, and the following section, we generalize the work of \cite{CK98}, and show a general algorithm for generating sub Hopf algebras of $\h_{rt}$ by applying the generalized natural growth operators. In particular, we show that all of $\h_{rt}$ can be generated by generalized natural growth operators.

\begin{thm}
The Hopf algebra of rooted trees, $\h_{rt}$, can be generated by
elements formed by repeated application of operators of the form
$N_t$ to the tree with a single vertex. \label{NtHopfalg}
\end{thm}

\begin{proof}
It is sufficient to show that any tree $t \in \h_{rt}$ can be written as a
finite linear combination of elements in the set \bas
\{N_{t_n}(\ldots(N_{t_1}(\bullet))\ldots )| n \in \N, t_i \textrm{
  rooted tree}\} \;.\eas

We proceed by induction on the fertility of the root vertex of $t$.
Any rooted tree $t$ with root fertility one can be written $t =
B_+(t_1)$,  \bas t = B_+(t_1) = N_{t_1}(\bullet) \;. \eas By induction,
suppose all trees with root fertility $m-1$, that is of the form
$B_+(t_1 \ldots t_{m-1})$, can be written as a linear combination of
trees formed by repeated application of general natural growth operators. By \eqref{NBrel}, the tree $t = B_+( t_1 \ldots  t_m)$ can be written \bas t =
N_{t_m}(B_+(t_1 \ldots  t_{m-1})) - \sum_{i=1}^{m-1} B_+ (t_1 \ldots
N_{t_m}(t_i) \ldots  t_{m-1}) \;. \eas

Since $B_+(t_1\ldots t_{m-1})$ and $B_+ (t_1 \ldots
N_{t_m}(t_i) \ldots t_{m-1})$ both have root fertility $m-1$, then
$t$ can be written as a linear combination of multiple natural growth
operators on the single vertex tree.
\end{proof}

\subsection{Sub Hopf algebras generated by natural growth}

One can generate sub Hopf algebras of $\h_{rt}$ by careful selection
of a family of trees by which to grow.

\begin{dfn}
Let $S$ be a set of rooted trees. This defines a set of natural growth
operators \bas \mcN(S) = \{N_t | t \in S \} \;.\eas Define $\A_S$ as
the sub algebra of $\h_{rt}$ generated by the trees in $S$ and the
repeated application of operators in $\mcN(S)$. That is \bas \A_S =
\Q[S][\{N_{t_n}(\ldots(N_{t_1}(s))\ldots )| n \in \N, t_i, s \in S\}]
\;.\eas
\end{dfn}

By construction, if $u \in \A_S$ and $t \in S$, then $N_t(u) \in \A_S$.

\begin{thm}
If $\Q[S] \in \h_{rt}$ is a Hopf algebra under $\Delta$, then so is $\A_S$.
\end{thm}

\begin{proof}
It is sufficient to show that \bas \Delta: \Q[S] \rightarrow \Q[S]
\otimes \Q[S] \quad \Rightarrow \quad \Delta: \A_S \rightarrow \A_S
\otimes \A_S \; . \eas Since, $\Q[S]$ and $\A_S$ are graded and
connected, this implies that they are Hopf algebras.

Consider $s,t \in S$, $N_t(s) \in \A_S$. By Theorem \ref{Ntcoprod},
\ba \Delta N_t(s) = (N_t\otimes \id)\Delta s +
\sum_{\begin{subarray}{c}c_t \textrm{ admis. cut} \\ \textrm{of }
    t\end{subarray}} (P_{c_t}(t) \otimes N_{R_{c_t}(t)}) \Delta s
\;. \label{recall}\ea If $\Q[S]$ is a Hopf algebra, $\Delta(s) \in
\Q[S] \otimes \Q[S]$. The first term in \eqref{Ntcoprod} is \bas
(N_t\otimes \id)\Delta s = \sum_{\begin{subarray}{c}c_s \textrm{
      admis. cut} \\ \textrm{of } s\end{subarray}} N_t(P_{c_s}(s))
\otimes R_{c_s}(s) \in \A_S \otimes \A_S. \eas The term
$N_t(P_{c_s}(s)) \in \A_S$ since $P_{c_s}(s) \in \Q[S]$.  Similarly,
the second term in equation \eqref{recall} can be written \bas
\sum_{\begin{subarray}{c}c_s \textrm{ admis. cut} \\ \textrm{of }
    s\end{subarray}} \sum_{\begin{subarray}{c}c_t \textrm{ admis. cut}
    \\ \textrm{of } t\end{subarray}} P_{c_t}(t) P_{c_s}(s) \otimes
N_{R_{c_t}(t)}(R_{c_s}(s)) \in \A_S \otimes \A_S\;. \eas Since the trees $t,
s \in S$,  the Hopf algebra structure of $\Q[S]$ implies that
$R_{c_s}(s)$ and $R_{c_t}(t)$ are as well. Therefore
$N_{R_{c_t}(t)}(R_{c_s}(s)) \in \A_S$.

For a set of trees $\{s,t_1 \ldots t_n\}$ write $u_{s,t_1\ldots t_n}
\in \A_S$, defined \bas u_{s,t_1\ldots t_n} = (
N_{t_n}(\ldots(N_{t_1}(s))\ldots )) \;.\eas To calculate the coproduct
of $u$ \bas \Delta (u_{s,t_1\ldots t_n}) = (N_{t_n}\otimes \id)
\Delta u_{s,t_1\ldots t_{n-1}} + \sum_{\begin{subarray}{c}c\textrm{
      admis. cut} \\ \textrm{of } t_n\end{subarray}} (P_c(t) \otimes
N_{R_c(t)}) \Delta (u_{s,t_1\ldots t_{n-1}}) \;. \eas By induction,
suppose that for all elements of the form $u_{s, t_1, \ldots t_k}$ for
$k <n$, and $s, t_i \in S$, \bas \Delta u_{s, t_1, \ldots t_k} \in
\A_S \otimes \A_S\;.\eas Then by similar arguments as above, \bas
\Delta (u_{s,t_1\ldots t_n}) \in \A_S \otimes \A_S\;. \eas Since the
$N_t$ are algebra homomorphisms, this extends to all $u \in \A_S$.
\end{proof}

\begin{eg}
Let $S = \{\bullet\}$. The polynomial algebra $\Q[S]$ is a Hopf algebra
  since \bas \Delta \bullet = 1 \otimes \bullet + \bullet \otimes 1
  \;. \eas The Hopf algebra $\A_\bullet$ is  $\h_{CK}$. \label{hCK}
\end{eg}

The rest of this section is devoted to computing the sub Hopf algebras
defined by corollas.

\begin{dfn}
Let $C_i$ be the tree with $i$ vertices, $i-1$ of which have no daughters. These are
called corollas. \end{dfn}

For instance, \bas C_1 = \xy \POS (0,0) *+{\bullet} *\cir{} \endxy
\quad ; \quad C_2 = \xy \POS(0,4) *+{\bullet} *\cir{} \POS(0,4)
\ar@{-} + (0, -8) \POS(0, -4) *+{\bullet} \endxy \quad ; \quad F_3 =
\xy \POS(0,4) *+{\bullet} *\cir{} \POS(0,4) \ar@{-} + (-4, -8)
\POS(0,4) \ar@{-} + (+4, -8) \POS(-4, -4) *+{\bullet} \POS(4, -4)
*+{\bullet} \endxy \quad ;\textrm{ and } \quad C_i = \underbrace{\xy
  \POS(0,4) *+{\bullet} *\cir{} \POS(0,4) \ar@{-} + (-6, -8) \POS(0,4)
  \ar@{-} + (+6, -8) \POS(-6, -4) *+{\bullet} \POS(6, -4) *+{\bullet}
  \POS(0, -4) *+{\ldots} \endxy}_{i-1 \text{ times}} \;. \eas The
function associated to each corolla is \bas \phi^j(C_i) =
\phi(\bullet)^{k_1}\cdots \phi(\bullet)^{k_{i-1}} (\D_{k_1\ldots
  k_{i-1}}\phi(\bullet)^j) \;. \eas

Let $S_k = \{C_i | i \leq k\}$. Using this notation, $\h_{CK} =
\A_{S_1}$.

\begin{thm}
The algebras $\A_{S_k}$ are  Hopf algebras.
\end{thm}

\begin{proof}
We need only check that $\Q[S_k]$ is a Hopf algebra. To see this,
notice that for $n \leq k$, \bas \Delta C_n = 1 \otimes C_n + C_n
\otimes 1 + \sum_{i= 1}^{n-1} {n-1\choose i}\bullet ^i \otimes
C_{n-i-1} \in \Q[S_k] \otimes \Q[S_k]\;. \eas
\end{proof}

\section{The Connes Moscovici Hopf algebras and their generalizations \label{CMpict}}

In \cite{CMos98}, the authors define a Hopf algebra $\h(n)$ on vector
fields over $F^+M$, for $M$ $n$-dimensional manifolds. This section
recalls their construction, and generalizes it. For a clear
explanation of Connes and Moscovici's construction of $\h(n)$, see
\cite{CMexplain} and \cite{RangipourMoscovici09}.

For a $n$-dimensional smooth manifold, $M$, let the coordinates \bas x^\mu : U \subset M
\rightarrow  \R^n\eas parametrize
a coordinate patch $U \subset M$. For $s \in \R^n$, write $y_i^\mu =
\frac{\D}{\D s_i} x^\mu$. Then $x^\mu$, $y^\mu_i$ form a set of local coordinates of
$TU$. In this paper, we consider only the case of $1$ dimensional manifolds, though exploring the combinatorial properties of higher dimensions poses an interesting question for future work.

Connes and Moscovici's are only interested in orientation preserving
diffeomorphisms on $M$, and thus only consider orientation preserving
frames $F^+(M)$. In the one dimensional case, this means they write
\ba y(s) = e^{z(s)} \label{expparam} \ea for some $z : \R \rightarrow M$. Their analysis involves the study of the curvature of $M$ and connections on $TF^+(U)$. However, curvature is a meaningless concept for one dimensional manifold. Instead of studying the standard Christoffel symbol on $M$, define a function \ba \Gamma(x) =
-\frac{1}{(\D_s x)^2} \D_s^2x \label{Gammadef}\ea to take its role in the one dimensional case. If $\psi$ is an orientation preserving local diffeomorphism on $M$ with $\textrm{Dom}(\psi) \subset U$,
then \bas \Gamma(\psi) = -\frac{1}{(\D_s \psi)^2} \D_s^2\psi \;. \eas The $\Gamma$ function transforms under coordinate change as Christoffel symbols do. Namely, \ba \Gamma(x)|_x =
\D_x\psi|_x\Gamma(\psi)|_{\psi(x)} + \frac{1}{\D_x\psi}
\D_x^2\psi|_x \label{Gamma} \;. \ea The authors then
define a connection $\mathfrak{g}(1)$ valued one form on
$F^+(U)$, \bas \omega = (y^{-1})(\D_y +
\Gamma y \D_x) \;,\eas  and a vector field over $F^+(U)$ \ba Y = y
dy \label{onedimY} \ea which generates the $GL^+(1, \R)$ action on $
F^+(U)$. The connection
defines a horizontal vector field over $F^+(U)$ \bas X =
y(\D_x - \Gamma y \D_y)
\;.\eas For $g \in C^\infty_c(U)$, using \eqref{Gammadef}, the action of the vector field $X$ simplifies as \ba X (g) = \ds g \label{onedimX}
\;.\ea

\begin{rem}
Compare the vector fields \eqref{onedimY} and \eqref{onedimX} to the vector fields defined in \eqref{phinatgrowth}. For $t = \bullet$, the vector field $X = \phi_{N_\bullet}$. For $t = 1$, $Y = \phi_{N_1} = \phi_Y$. It is this observation that motivates the analysis in the rest of this paper.
\end{rem}

The vector fields $X$
and $Y$ act on the algebra generated by compactly supported functions
on $M$ crossed with the pseudo group of orientation preserving
diffeomorphisms of $M$.

\begin{dfn}
Let $\textrm{Diff}^+(M)$ be the pseudo group of orientation preserving
local diffeomorphisms on $M$. Define a group \bas \A = C_c^\infty
(F^+(M))\rtimes \textrm{Diff}^+(M) \eas that is the semi-direct product
of compactly supported smooth functions on $F^+(M)$ with orientation
preserving diffeomorphisms of $M$. \end{dfn}

In order to explore the action of $X$ and $Y$ on this algebra, and to generalize them, we first review the construction for $\A$ and some Riemannian geometry, following \cite{CK98} and \cite{CMexplain}. While much of what follows can be generalized to a multi-dimensional manifold, $M$, we continue our exposition in the one dimensional case, as that is all we study in detail.

Any $\psi \in \textrm{Diff}^+(M)$ lifts to a diffeomorphism of the
frame bundle $\tilde\psi \in \textrm{Diff}^+(F^+M)$ \bas \tilde \psi :
\{x, y\} \rightarrow \{\tilde x := \psi(x), \tilde
y: = y \D_x \psi \} \;. \eas The group $\A$ is
generated by the monomials \bas f U_\psi^* \in \A \quad f \in
C_c^\infty(\textrm{Dom}\tilde \psi) \quad \psi \in
\textrm{Diff}^+(U)\;.\eas The $*$ in the monomial corresponds to the
contravariant multiplication \bas U^*_{\psi_1}U^*_{\psi_2} =
U^*_{\psi_2 \circ \psi_1} \;.\eas Multiplication in the group is given by
\bas (f_1 U_{\psi_1}^*)(f_2 U_{\psi_2}^*) = f_1 U_{\psi_1}^*
f_2U_{\psi_1}^{*-1} U_{\psi_1}^* U_{\psi_2}^* = f_1 \cdot (f_2 \circ
\tilde\psi_1) U_{\psi_1\psi_2}^* \eas where $\cdot$ corresponds to
point-wise multiplication. Writing $U_\psi^{* -1} = U_{\psi^-1}^*$,
commutation with $U_\psi^*$ \ba U_\psi^* f U_{\psi^-1}^* = f(p)
\circ \tilde \psi(p) \; \label{commutator} \ea results in changing
the point of evaluation from $p \in F^+(U)$ to $\tilde \psi(p)$. The domain
of $f_1 \cdot (f_2 \circ \tilde\psi_2)$ is \bas \textrm{Dom}(\tilde
\psi_1) \cap \tilde \psi_1^{-1}(\textrm{Dom} \tilde \psi_2) \;. \eas

The actions of $X_i$ and $Y$ on the monomials $fU_\psi^*$ are given by
\bas X_i(fU_\psi^*) = (X_if)U_\psi^* \quad \quad Y(fU_\psi^*) =
(Yf)U_\psi^* \;.\eas Before recalling their actions on products, we
recall the pushforwards of these vector fields. For a point $p \in
F^+(U)$, and $\tilde p = \tilde \psi (p)$. Then \ba \tilde \psi_*
Y|_{\tilde p} f = \tilde y \frac{\D f(\tilde p) }{\D \tilde
  y} = y|_{\tilde p} \frac{\D f(p)}{\D y}|_{\tilde
  p} = Y|_{\tilde p} f \label{Ypullback}\;.\ea That is, the vector
field $Y$ is invariant under orientation preserving
diffeomorphisms. On the other hand, \bas \tilde \psi_* X|_{\tilde p} f =
\frac{\D \;}{\D s}(f\circ \tilde \psi(p)) \neq \frac{\D\;}{\D
  s_i}|_{\tilde p}(f) \;. \eas Let $\tilde \Gamma$ be the
Christoffel symbol under the change of coordinates on $M$ induced by the orientation preserving diffeomorphism $\psi$. Then \ba \tilde \psi_* X|_{\tilde p} = \tilde
y(\frac{\D \;}{\D \tilde x} - \tilde \Gamma \tilde y \D \tilde y) \;.\label{Xpush}\ea

Any vector field $V$ acting on a product of monomials
$(f_1U_{\psi_1}^*)(f_2U_{\psi_1}^*)$, gives \cite{CMexplain} \bml
V|_p(f_1U_{\psi_1}^*)(f_2U_{\psi_1}^*) = V|_p(f_1)\cdot (f_2 \circ
\tilde \psi_1(p)) U_{\psi_2\psi_1}^* + f_1(p)\cdot \tilde
\psi_{1*}V|_{\tilde \psi_1(p)}f_2 U_{\psi_2\psi_1}^* =
\\ (V(f_1) U_{\psi_1}^*) (f_2 U_{\psi_2}^*) + (f_1 U_{\psi_1}^*)\cdot
U_{\psi_1^{-1}}^* \tilde \psi_{1*}V|_{\tilde \psi_1(p)}f_2
U_{\psi_1}^*U_{\psi_2}^* = \\ (V(f_1) U_{\psi_1}^*) (f_2 U_{\psi_2}^*) +
(f_1 U_{\psi_1}^*) ( \tilde \psi_{1*}V|_pf_2 U_{\psi_2}^*)
\;. \label{genvfield}\eml By \eqref{Ypullback}, since $\tilde
\psi_{1*}Y|_p = Y|_p $, \ba Y|_p(f_1 U_{\psi_1}^*)( f_2
U_{\psi_2}^*) = (Y|_p (f_1U_{\psi_1}^*)) f_2U_{\psi_2}^* + (f_1
U_{\psi_1}^*)(Y|_p(f_2U_{\psi_2}^*)) \;.\label{Yprod} \ea Applying
\eqref{genvfield} to $X_i$ gives \bas
X|_p(f_1U_{\psi_1}^*)(f_2U_{\psi_1}^*) = (X|_p(f_1)
U_{\psi_1}^*)(f_2U_{\psi_2}^*) + (f_1U_{\psi_1}^*)((\tilde
\psi_{1*}X)|_pf_2 U_{\psi_2}^*) \;. \eas Rewrite this \ba
(X(f_1)U_{\psi_1}^*)(f_2U_{\psi_2}^*) + (f_1U_{\psi_1}^*)
(X(f_2)U_{\psi_2}^*) + (f_1U_{\psi_1}^*)((\tilde \psi_{1*}X -
X)(f_2)U_{\psi_2}^*) \;.\label{Xprod} \ea By \eqref{Xpush}, \bas
(\tilde \psi_{1*}X - X)|_p = \Gamma|_p
(y Y)|_p - \tilde \Gamma|_{p} (\tilde y\tilde Y)|_p  \;. \eas

\begin{dfn}
Write $\gamma|_p(\psi_1) Y = (\tilde \psi_{1*}X - X)|_p$. That is, \bas \gamma|_p(\psi) = y \Gamma |_p - \tilde y \tilde Y|p \;.\eas
Define a linear operator on $\A$, $\delta_1$, such that \bas
\delta_1(f_1 U_{\psi_1}^*) = (\gamma|_p(\psi_1)f_1
U_{\psi_1}^*) \;.\eas
\label{delta1def}\end{dfn}

For $a, b \in \A$, equation \eqref{Xprod} gives \bas X(ab) = X(a)
b + aX(b) + \delta_1(a) Y(b) \;.\eas Connes and Moscovici
show that \bas \delta(a b) = \delta(a)
\delta(b)\;.\eas Writing $\psi' = \D_x
\psi$, and using equation \eqref{Gamma}, at the point $q = \tilde \psi
(p)$, \ba \gamma(\psi)|_q = y \Gamma(x)|_q - y \Gamma(x)|_p +
\frac{1}{\psi'}\frac{d \psi'}{ds}|_p \;. \label{gammadef} \ea

\begin{dfn}
Let $\h_{CK}$ be the algebra defined \bas \h_{CK} =
\Q[\{\delta_i| i \in \N\}] \;.\eas Let $\mathfrak{g}_1$ be the Lie algebra generated by $X$ and $Y$.
\end{dfn}

In $\mathfrak{g}_1$, there is the
  commutation relation \bas [Y, X] = X \;. \eas There is a right
  coaction in the bicrossed product ${\h_{CK}} \acl
  \mathcal{U}(\mathfrak{g}_1)$ given by \bas Y \rightarrow Y \otimes
  1 \quad ; \quad X \rightarrow X \otimes 1 + \delta_1 \otimes Y \;,
  \eas and a left coaction \bas X \delta_n = \delta_{n+1}, \quad ;
  \quad Y \delta_n = n \delta_{n} \;.\eas This algebra $\h_{CK}$ is actually a Hopf algebra, as shown in \cite{CMos98}.
The bicrossed product of these two Hopf algebras, gives the Hopf algebra of interest \bas \h(1) = \h_{CK} \acl {\mathcal{U}(\mathfrak{g}_1)}\;.\eas
This has all been done in the one dimensional case. See \cite{RangipourMoscovici09} for a generalization of this construction to the $n$ dimensional case.


In \cite{CK98}, the authors show that
$\h_{CK}$ is isomorphic to the sub Hopf algebra of rooted trees
formed by applying the natural growth operator, $N_\bullet$, to the
tree $\bullet$ from example \ref{hCK}. The action of the vector field
$X$ on $\h_{CK}$ corresponds to the natural growth by a single
vertex. The action of the vector field $Y$ on $\h_{CK}$ corresponds to
the grading function, or, by \eqref{natgrowbyone} natural growth by $1$.

\subsection{The new Hopf algebra $\h_{rt}(1)$}

In this section, we enlarge $\h(1)$ to a Hopf algebra \ba \h_{rt}(1)
:= {\h_{rt}} \acl {\mathcal{U}(\mathfrak{g}_{rt})}\; , \label{Hrtdef}\ea where
  $\mathfrak{g}_{rt}$ is a Lie algebra generated by vector fields of
  the form $X_t$, for $t$ a rooted tree. We spend this section
  defining this Lie algebra, and showing that the left action of $X_t$
  on linear operators $\delta_{t'}$, for some $t' \in \h_{rt}$
  corresponds to natural growth of $t'$ by $t$.

Since the vector field $X$ in $\h(1)$ acts on $\delta_i$ as natural
growth by $\bullet$, and $Y(\delta_1) = \delta_1$, we rename these
maps to indicate the trees they correspond to in
$\h_{rt}$. Specifically, define \bas X_\bullet := X \quad ; \quad
\delta_\bullet:=\delta_1 \;. \eas By this notation, define the
function $\gamma(\psi) = \gamma_\bullet(\psi)$. Using the coordinates
introduced in \eqref{expparam} and equation \eqref{onedimX} we
rewrite $X_\bullet$ in terms of the independent coordinates $x,
z$. Specifically, $z = \log \frac{dx}{ds}$, $Y = \D_z$, and \bas
X_\bullet = e^z \D_x - e^z\Gamma(x) \D_z \;.  \eas Notice that \ba
\frac{dx}{ds} = e^z \quad ; \quad \frac{dz}{ds} =- e^z\Gamma(x)
\;.\label{FMODE}\ea For the remainder of this paper, this defines the
system of differential equations of interest. We apply the notation
developed in section \ref{NB} to the equations in \eqref{FMODE},
with the index $i \in \{x, z\}$. In this context, rewrite \bas
X_\bullet = \phi(\bullet)^i \D_i \;. \eas The pushforward of
$X_\bullet$ under an orientation preserving diffeomorphism $\tilde
\psi \in \textrm{Diff}^+(F^+M)$ gives the expression \bas \tilde \psi_*
X_\bullet = \phi^x(\bullet) \D_x + (\phi^z(\bullet) + \ds \; \log
\psi') \D_z \;. \eas Similarly, for $q = \tilde\psi (p)$, $\gamma_\bullet (\psi)$ becomes \bas
\gamma_\bullet (\psi)|_q = \phi^z(\bullet)|_p +
(\ds \log \psi')|_p - \phi^z(\bullet)|_q
\;. \eas

\begin{dfn}
Define a function \bas \gamma_t(\psi) = \phi_t(\gamma_\bullet(\psi))
\;, \eas according to Definition \ref{phimapgen}. This defines a linear
operator $\delta_t$ on $\A$ \bas \delta_t (f U_\psi) = \gamma_t(\psi)
f U_\psi \;. \eas
\label{deltatdef}
\end{dfn}

For a forest $t t'$ of rooted trees, we write \bas \delta_{tt'} :=
\delta_t \delta_{t'} \;.\eas

\begin{lem}
Consider the operator \bas \phi_t: C_c^\infty (F^+M) \rightarrow
C_c^\infty (F^+M) \;. \eas Write $t = B_+(t_1 \ldots t_n)$. The
pushforward of $\phi_t$ by $\tilde \psi \in \textrm{Diff}^+(F^+M)$ is
given by \bas \tilde \psi_* \phi_t = \prod_{j = 1}^n \phi_{t_j}
(\tilde \psi_*\phi^{i_j} (\bullet)) \tilde \psi_*(\prod_{j = 1}^n \D_{i_j}) \;.\eas
\label{pushforward} \end{lem}

\begin{proof}
Let $\tilde \psi \in \textrm{Diff}^+(F^+M)$. Recall that $\tilde\psi =
(\psi, \log (\ds \psi) )$, for $\psi \in \textrm{Diff}^+(M)$. Write
$\tilde\phi$ as the map from $\h_{rt}$ to $C_c^\infty(F^+M)$ defined
by coordinates defined by $\tilde \psi$. Define $\tilde \phi(t)$ accordingly. For $t = \bullet$, this is
$\tilde \psi_*\phi^i(\bullet)$, \bas \tilde \psi_*\phi^i(\bullet) =
\ds \tilde\psi^i = \phi^j(\bullet) \D_j \tilde\psi^i = \tilde
\phi^i(\bullet) \;. \eas Since $\tilde \phi^i(\bullet) \frac{\D}{\D
  \tilde \psi^i} = \phi^i(\bullet) \D_i$, we get \bas \tilde \phi^i(t)
= \prod_{j = 1}^n \phi^{i_j}(t_j) \prod_{j = 1}^n\D_{i_j} \tilde
\psi_*\phi^i(\bullet) \eas for $t = B_+(t_1 \ldots t_n)$. In other
words \ba \tilde \phi^i(t) = \phi_t (\tilde
\psi_*\phi^i(\bullet))\;. \label{tildephi}\ea

The pushforward of $\phi_t$ is \bas \tilde \psi_* \phi_t f = \phi_t (f
\circ \tilde \psi) \;. \eas Applying this to equation \eqref{Bphinotation}
gives \bas \tilde \psi_* \phi_t = \prod_{j = 1}^n \phi^{i_j}(t_j)
\prod_{j = 1}^n \D_{i_j} (f \circ \tilde \psi) \eas which can be rewritten
using \eqref{tildephi} as \bas \prod_{j = 1}^n \tilde \phi(t_j)
\prod_{j = 1}^n \frac{\D}{\D \tilde \psi^{i_j}} f = \prod_{j = 1}^n
\phi_{t_j} (\tilde \psi_*\phi^i(\bullet)) (\prod_{j = 1}^n
\frac{\D}{\D \tilde \psi^{i_j}}) f \;.  \eas
\end{proof}

The pushforward of $\phi_t$ leads to a useful expression for
calculating its action for a product of monomials in $\A$.

\begin{lem}
The pushforward \bas \tilde \psi_* \phi_t = \sum_{\begin{subarray}{c}
    c \textrm{ admis.  cut}\\ c \neq \textrm{ full}
\end{subarray}} \gamma_{P_c(t)}(\psi) \phi_{R_c(t)}|_{\tilde \psi} \;.
\eas For $R_c(t) = 1$, we write $\phi_1 = Y$.\label{coprodcontrib}
\end{lem}

\begin{proof}
Write $t = B_+(t_1\ldots t_n)$. By Lemma \ref{pushforward} write
\bas \tilde \psi_*\phi_t = \left( \prod_{j=1}^k \phi_{t_j} (\tilde
\psi_* \phi^{i_j}(\bullet)) \prod_{j=1}^k \tilde \psi_*\D_{i_j}\right)
\;. \eas For subsets $I \subseteq \{1,
\ldots k\}$, \ba \tilde \psi_* \phi_t =  \sum_{I} \left( \prod_{j
  \not \in I} \phi_{t_j} (\tilde \psi_* \phi(\bullet) -
\phi(\bullet)|_{\tilde \psi})^{i_j} \prod_{l \in I} \phi_{t_l}(
\phi(\bullet)^{i_l}\circ\tilde \psi) \right)\prod_{j \not \in I}
(\tilde \psi_*\D-\D|_{\tilde \psi})_{i_j}\prod_{l \in I} \D_{i_j}|_{\tilde \psi}
 \label{expand} \;. \ea Making the substitution \bas
(\tilde \psi_* \phi(\bullet) - \phi (\bullet))^{i_j} (\tilde
 \psi_*\D-\D)_{i_j}= \gamma_{\bullet}(\psi) Y\; \eas into any of the
 summands of \eqref{expand} corresponds to making an admissible cut of
 $t$ with pruned forest $\prod_{j\not \in I} t_j$. In the expression
 for $\tilde \psi_* \phi_t$, these appear as $\prod_{j\not \in I}
 \phi_{t_j}(\gamma_{\bullet}(\psi))$. For a fixed set $I$, and a fixed
 $n \in I$, write the tree $t_n$ in \eqref{expand} $t_n =
 B_+(t_{n_1}\ldots t_{n_r})$. For subsets $I' \subseteq \{1 \ldots
 r\}$, use the definition of a pushforward and equation \eqref{expand}
 to write the function \ba \phi_{t_n}( \phi^{i_n}(\bullet)\circ \tilde
 \psi) = \tilde \psi_* \phi_{t_n}( \phi^{i_n}(\bullet)) = \nonumber
 \\ \sum_{I'} \left( \prod_{j\not \in I'} \phi_{t_{n_j}} (\tilde
 \psi_* \phi(\bullet) - \phi(\bullet)|_{\tilde \psi})^{i_j} \prod_{l
   \in I'} \phi_{t_{n_l}} (\phi(\bullet)^{i_l}\circ \tilde \psi)
 \right)\prod_{j\not \in I'} (\tilde \psi_*\D-\D|_{\tilde
   \psi})_{i_j}\prod_{l\in I'} \D_{i_l}|_{\tilde \psi}
 (\phi^{i_n}(\bullet))\;. \label{substitute}\ea Substituting this
 expression into \eqref{expand} corresponds to also taking an admissible
 cut of $t_n$. Therefore, the summands of $\tilde \psi_* \phi_t$, as
 given by equations \eqref{expand} and \eqref{substitute} correspond
 to non-full admissible cuts of $t$. For each such cut, $c$, the
 function associated to the pruned forest by the map $\phi$ is of the
 form \bas \phi_{R_c(t)} \gamma_{\bullet} (\psi) =
 \gamma_{P_c(t)}(\psi) \;,\eas and the root tree is \bas
 \phi_{R_c(t)}|_{\tilde \psi} \;.\eas
\end{proof}

We use this result to calculate the coproduct of $\delta_t$.

\begin{thm}
The coproduct \bas \Delta \delta_t = \sum_{c \textrm{ admis.
    cut}}\delta_{P_c(t)} \otimes \delta_{R_c(t)} \;.\eas
\end{thm}

\begin{proof}
By definition \bas \Delta \delta_t (f U_\psi)( g U_\eta) =
\gamma_t(\eta \circ \psi) f\cdot (g\circ \psi) U_{\eta\circ \psi}
\;, \eas and \bas \gamma_t(\eta \circ \psi) = \phi_t
(\gamma_\bullet(\eta \circ \psi)) \;.\eas From the coproduct of
$\delta_\bullet$, \bas \gamma_\bullet (\eta \circ \psi) =
\gamma_\bullet (\psi)+ \gamma_\bullet (\eta) \circ \tilde \psi \;,
\eas so \bas \gamma_t(\eta \circ \psi)= \phi_t (\gamma_\bullet
(\psi)+ \gamma_\bullet(\eta)\circ \tilde \psi ) = \gamma_t(\psi) +
\tilde \psi_*\gamma_t(\eta)\;. \eas Rewriting $\tilde
\psi_*\gamma_t(\eta) = \tilde \psi_* \phi_t (\gamma_\bullet(\eta))$,
Lemma \ref{coprodcontrib} gives \bas \gamma_t(\eta \circ \psi) =
\gamma_t(\psi) + \sum_{\begin{subarray}{c}c \textrm{ admis. cut} \\ c
    \neq \textrm{ full} \end{subarray}}\gamma_{P_c(t)}(\psi)
\gamma_{R_c(t)}(\eta)\circ \tilde \psi \;. \eas Therefore \bas
\delta_t (f U_\psi)( g U_\eta) = \sum_{c \textrm{
    admis. cut}}(\gamma_{P_c(t)}(\psi)f U_\psi)(
\gamma_{R_c(t)}(\eta)g U_\eta) \;.\eas
\end{proof}

We use these linear operators to define a family of vector fields
corresponding to natural growth.

\begin{dfn}
Define a family of vector fields \bas X_t = \phi(t)^i \D_i \;. \eas
These act on $\A$ by $X_t(f U_\psi) = (X_t f)U_\psi$.
\end{dfn}

\begin{rem}Notice that $X_t = \phi_{N_t}$ as defined in equation
\eqref{phinatgrowth}. By equation \eqref{NBrel2}, one can write $X_t =
\phi_{B_+(t)} $. \label{Xt}\end{rem}

\begin{thm}
The coproduct of the vector fields $X_t$ is the same as the coproduct
of the natural growth operator on $\h_{rt}$. Namely, \bas \Delta X_t =
X_t \otimes 1 + \sum_{c \textrm{ admis. cut of } t } \delta_{P_c(t)} \otimes X_{R_c(t)}
\;. \eas As with natural growth, we identify $Y = X_1$.
\end{thm}

\begin{proof}
The coproduct of $X_t$ is calculated by evaluating \bas X_t (f
U_\psi)( g U_\eta) = (X_t f)\cdot (g\circ \psi) U_{\eta\psi} + f \cdot
(\tilde \psi_* X_t g) U_{\eta\psi}\;.\eas Write $X_t = \phi_{B_+(t)}$,
as in Remark \ref{Xt}. Lemma \ref{coprodcontrib} gives \bas \tilde
\psi_* X_t = \sum_{\begin{subarray}{c} c \textrm{ admis.  cuts of }
    B_+(t)\\ c \neq \textrm{ full}\end{subarray}}
\gamma_{P_c(B_+(t))}(\psi) \phi_{R_c(B_+(t))}|_{\tilde \psi} \;.\eas

The edge set of $B_+(t)$ is given by $E(B_+(t)) = E(t) \cup e$. A non-full
admissible cut of $B_+(t)$ is either an admissible cut of $t$, or $c =
e$. For any non-full admissible cut of $B_+(t)$, $c$, $R_c(B_+(t)) =
B_+(R_c(t))$. If $c = e$, it is the full cut of $t$, and $R_c(t) = 1$.
Therefore, \bas \psi_* X_t = \sum_{c \textrm{ admis.  cuts}}
\gamma_{P_c(t)} X_{R_c(t)}|_{\tilde \psi} \;. \eas
\end{proof}

This is of the same form as the coproduct of the natural growth operator $N_t$
from Theorem \ref{Ntcoprod}.

\begin{dfn}
Let $\mathfrak{g}_{rt}$ be the Lie algebra spanned by the element $Y$,
$\{X_t | t \textrm{ rooted tree}\}$ \end{dfn}

It remains to check the commutation relations on $\mathfrak{g}_{rt}$.

\begin{thm}
The following commutation relations hold. Let $t$ and $t'$ be rooted
trees. \begin{enumerate} \item $ [Y, X_t] = Y(t) X_t $ \item $[X_t, X_{t'}] = X_{N_t(B_+(t'))} - X_{N_{t'}(B_+(t))} \quad ; \quad [\delta_t,
  \delta_{t'}] = 0 $ . \end{enumerate}

\end{thm}

\begin{proof}
Recall that $x$ and $z$ are independent variables, and that
$\Gamma(x)$ is only a function of $x$. Both $\phi^x(\bullet)$ and
$\phi^z(\bullet)$ have a factor of $e^z$. Therefore, $Y(\phi^i(t)) =
Y(t) \phi^i(t) $. Writing \bas [Y, X_t](f) = Y(\phi^i(t))\D_i(f
U_\psi) = |V(t)| \phi^i(t) \D_i(f U_\psi) = Y(t) X_t (f U_\psi)
\;. \eas To calculate the relation \bas [X_t, X_{t'}] =
X_{N_t(B_+(t'))} - X_{N_{t'}(B_+(t))} \eas use equation \eqref{Bphinotation} to write $
X_{t'} = \phi(B_+(t')$. Then, by Lemma
\ref{natgrowthgen} \bas X_tX_t' = \phi_{N_t(B_+(t'))} = X_{N_t(B_+(t'))} \;. \eas
\end{proof}

Next, we show the left action of $\mathfrak{g}_{rt}$ on $\h_{rt}$.

\begin{thm}
The Lie algebra $\mathfrak{g}_{rt}$ acts on the Hopf algebra $\h_{rt}$ by
\bas X_t( \delta_{t'}) =
\delta_{N_t(t')} \quad ; \quad [Y, \delta_t] = Y(t) \delta_t \eas
\end{thm}

\begin{proof}
The proof of this theorem is similar to that of the previous theorem.  \bas Y(\delta_t) (f U_\psi)= Y(t) \phi_t (\gamma_t(\psi)) = Y(t)
 \delta_t (f U_\psi) \;. \eas To calculate \bas [X_t, \delta_{t'}](f
 U_\psi) = X_t (\phi_{t'}(\gamma_\bullet (\psi))) f U_\psi \;, \eas
 use Lemma \ref{natgrowthgen} to write \bas [X_t, \delta_{t'}](f
 U_\psi) = \phi_{N_t(t')} (\gamma_\bullet (\psi))(f U_\psi) =
 \delta_{N_t(t')} (f U_\psi) \; .\eas
\end{proof}

Finally, it is worth noting that the operators $\delta_t$ can be
written in terms of a series of commutators of operators of the form
$X_{t'}$ with $\delta_\bullet$.

\begin{thm}
The Hopf algebra $\h_{rt}(1)$, defined in \eqref{Hrtdef}, is isomorphic to the Hopf
algebra defined by the bi-crossed product \bas \h_{rt}(1) = {\h_{CK}} \acl
{\mathcal{U}(\mathfrak{g}_{rt})} \;. \eas
\end{thm}

\begin{proof}
This is a corollary of Theorem \ref{NtHopfalg}. Following the
arguments presented there, write $t = B_+(t_1 \ldots t_n)$. Assume by
induction that for all $t'$ with root fertility $<n$, $\delta_{t'}$
can be expressed in terms of linear combinations of commutators of
$\delta_\bullet$ with a set $\{X_{t_i} \}$. Then \bas \delta_t (f
U_\psi) = X_{t_n} (\delta_{B_+(t_1\ldots t_{n-1})})(f U_\psi) -
\sum_{i=1}^{n-1}\delta_{B_+(t_1 \ldots N_{t_n}(t_i) \ldots t_{n-1})}(f
U_\psi) \;. \eas
\end{proof}

The crucial adjustment needed in order to generalize the Hopf
algebra $\h(1)$ is that $\Gamma(x) \neq 0$. If $\Gamma(x) = 0$, then
for any non-trivial rooted tree $t \neq \bullet$, $\phi^i(t) =
0$. Since $\phi^z(\bullet) = -e^z \Gamma(x)$, $\Gamma(x) = 0$ implies
that the only non-zero term of the form $\D_i \phi^j(\bullet)$ is
$\D_z \phi^x(\bullet)$. Therefore, $\Gamma(x) = 0$ implies that for
any $t \neq \bullet$, only $\phi^x(t)$ is possibly non-zero. However,
for $t = B_+(t_1\ldots t_n)$, one sees that \bas \phi^x(t) =
\prod_{i=1}^n \phi^z(t_i) \D_z^n (\phi^x(\bullet)) = 0 . \eas
Therefore, any generalization of this type for the algebras $\h(n)$,
for $n \geq 1$ must incorporate the curvature of the base manifold
$M$, as the commutator $[X_i, X_j] = R^k_{lij}Y^l_k$, for $R^k_{lij}$
the curvature of $M$, and a torsion free connection $\omega^i_j$ on
$M$. Towards this goal, work by \cite{LMK} shows that planar rooted
trees arise naturally in the case of flat connections with constant
torsion.

\section*{Acknowledgments}

The authors would like to thank Karen Yeats for several useful
discussions during the early stages of this project. We are also
grateful to the Caltech SURF program for partial funding for this
project.

\bibliographystyle{amsplain}
\bibliography{Bibliography}{}
\end{document}